\documentclass[conference,11pt]{IEEEtran}
\IEEEoverridecommandlockouts
\usepackage{cite}
\usepackage{amsmath,amssymb,amsfonts}

\usepackage{graphicx}
\usepackage{textcomp}
\usepackage{xcolor}

\usepackage{bbm}
\usepackage{url}
\usepackage{subfigure}
\usepackage{pifont}
\usepackage{setspace}
\usepackage{multirow}
\usepackage{cases}
\usepackage{amsthm}

\usepackage{booktabs}
\usepackage[T1]{fontenc}
\usepackage{comment}
\usepackage{multicol}
\def\BibTeX{{\rm B\kern-.05em{\sc i\kern-.025em b}\kern-.08em
    T\kern-.1667em\lower.7ex\hbox{E}\kern-.125emX}}

\definecolor{orange}{rgb}{1, .36, .08}

\newcommand{\fakeparagraphnospace}[1]{\noindent\textbf{#1.}}

\newtheorem{definition}{\textbf{Definition}}
\newtheorem{theorem}{\textbf{Theorem}}
\newtheorem{corollary}{Corollary}[theorem]
\newtheorem{lemma}[theorem]{Lemma}

\renewcommand{\baselinestretch}{0.984} 

\usepackage[all]{background}
\usepackage{stackengine}
\setstackEOL{\\}
\setstackgap{L}{\normalbaselineskip}
\SetBgContents{\color{blue}{\tiny \Longstack{PREPRINT - accepted at IEEE International Symposium on Hardware Oriented Security and Trust (HOST), 2021}}}
\SetBgPosition{4.5cm,1cm}
\SetBgOpacity{1.0}
\SetBgAngle{0}
\SetBgScale{1.8}
\usepackage{algorithmic,algorithm}
\begin{document}

\title{Fun-SAT: Functional Corruptibility-Guided SAT-Based Attack on Sequential Logic Encryption
\thanks{$^1$Yinghua Hu and Yuke Zhang contributed equally to this work.}
}

\author{
Yinghua Hu,$^1$ Yuke Zhang,$^1$ Kaixin Yang, Dake Chen, Peter A. Beerel, and Pierluigi Nuzzo\\ 
\small Department of Electrical and Computer Engineering, University of Southern California, Los Angeles, CA, USA \\ \{yinghuah, yukezhan, kaixinya, dakechen, pabeerel, nuzzo\}@usc.edu\\[0ex]
}

\maketitle

\begin{abstract}
The SAT attack has shown to be efficient against most combinational logic encryption methods. It can be extended to attack sequential logic encryption techniques by leveraging circuit unrolling and model checking methods. However, with no guidance on the number of times that a circuit needs to be unrolled to find the correct key, the attack tends to solve many time-consuming Boolean satisfiability (SAT) and model checking problems, which can significantly hamper its efficiency. In this paper, we introduce \texttt{Fun-SAT}, a functional corruptibility-guided SAT-based attack that can significantly decrease the SAT solving and model checking time of a SAT-based attack on sequential encryption by efficiently estimating the minimum required number of circuit unrollings. \texttt{Fun-SAT} relies on a notion of functional corruptibility for encrypted sequential circuits and its relationship with the required number of circuit unrollings in a SAT-based attack. 
Numerical results show that \texttt{Fun-SAT} can be, on average, $\boldsymbol{90\times}$ faster than previous attacks against state-of-the-art encryption methods, when both attacks successfully complete before a one-day time-out. Moreover, \texttt{Fun-SAT} completes before the time-out on many more circuits.
\end{abstract}

\begin{IEEEkeywords}
Logic Encryption, Hardware Security, SAT-Based Attack
\end{IEEEkeywords}

\section{Introduction}\label{sec:intro}

Integrated circuits (ICs) are often regarded as the root of trust  of modern Internet-of-Things applications. However, their integrity and confidentiality can be seriously compromised by reverse engineering, among other threats, conducted by malicious parties in the supply chain.  Various methods have been proposed to mitigate this threat, such as split manufacturing~\cite{xiao2015efficient}, gate camouflaging~\cite{yasin2016camoperturb}, and logic encryption~\cite{yasin2017evolution}. Among these,  
logic encryption has gained significant attention over the past decade. 

A class of logic encryption methods, namely, \emph{combinational logic encryption}~\cite{roy2010ending,rajendran2013fault, yasin2016improving,chowdhury2021enhancing} aims to insert programmable elements and extra key ports into a portion of a circuit's combinational logic, so that the correct functionality can only be accessed by applying the correct key.
Models and metrics to help quantify the security and overhead of combinational logic encryption~\cite{vivek2019system,hu2019models} have been recently proposed and, consequently, security-oriented design tools~\cite{patnaik2018best,hu2021risk, mohan2021hardware} have started to appear.
On the other hand, \emph{sequential logic encryption}~\cite{chakraborty2009harpoon,desai2013interlocking} aims to create new states, marked as encrypted states, and modify the transitions in the circuit finite state machine. 
When powered on, the circuit is configured to be in an encrypted state 
and a predefined key sequence must be provided at the input ports before entering the true reset state. 

A growing number of attacks have targeted combinational encryption methods over the years, the most notable being the \emph{SAT attack}~\cite{subramanyan2015evaluating}, based on Boolean satisfiability (SAT) solving.  
Given an encrypted netlist and a functional chip, i.e., an \emph{oracle} providing the correct input/output response, the SAT attack searches for the correct key by solving a series of SAT problems that efficiently eliminate wrong keys. 
Because it assumes that the internal state is part of the input/output response, 
the SAT attack cannot be directly applied to a sequential circuit 
when the internal state is not scanned or is protected using a secure scan chain~\cite{wang2017secure}.  
However, in this case, a \emph{SAT-based attack}~\cite{shamsi2019kc2, el2017reverse} can still be developed 
by unrolling the sequential circuit 
to form a larger combinational circuit that represents the behavior of the original circuit over a number of clock cycles, and by applying the SAT attack to the unrolled version. A similar method can be used to formulate SAT-based attacks against sequential logic encryption~\cite{meade2017revisit,hu2020sanscrypt,hu2020sanscrypt_extended}.

A major challenge for these unrolling-based attacks on sequential circuits stems from the potentially large number of cycles a circuit needs to be unrolled in order to find the correct key, hence the growing size of the combinational circuits that need to be analyzed. To limit the size of the unrolled circuit, the attack usually starts by performing a SAT attack over a small number of unrollings. 
However, once the SAT attack terminates successfully, 
there is no guarantee that the set of candidate keys, obtained by matching the oracle response over a bounded time-horizon, will also match the circuit response for longer horizons. 
A model checking problem is then cast to verify the correctness of all the candidate keys. The attack gradually increases the unrolling depth and performs the SAT attack on the unrolled circuits 
until all the spurious keys are pruned out and the remaining keys are proven correct. Solving multiple instances of model checking and SAT attacks for growing unrolling depths can be expensive and drastically affect the feasibility of the overall attack.

This paper shows that the efficiency of SAT-based attacks on sequential logic encryption can be significantly improved by introducing an effective method to estimate the minimum number of unrollings needed to find the correct key. We propose \texttt{Fun-SAT}, a functional corruptibility-guided SAT-based attack, which relies on a notion of functional corruptibility for a sequential circuit to reduce both the SAT-attack effort spent to produce candidate key sets and the model checking effort spent to prove that a candidate key set is correct. 
Our contributions can be summarized as follows: 
\begin{itemize}
    \item We introduce a notion of functional corruptibility for sequential circuits and characterize 
    its relation with the set of wrong keys that are pruned out by a SAT attack at each unrolling depth.
    \item We develop \texttt{Fun-SAT}, an attack to sequential logic encryption that leverages functional corruptibility 
    to significantly reduce the overall execution time. 
    \item We evaluate \texttt{Fun-SAT} on two state-of-the-art sequential logic encryption methods
showing that, on average, it can be $90\times$ faster than previous SAT-based attacks. Only $0.7\%$ of the experiments timed out after one day, compared with $19\%$ for the previous attack. 
\end{itemize}

The remainder of the paper is organized as follows. Section~\ref{sec:background} introduces two state-of-the-art sequential logic encryption methods and
reviews the mechanism of the SAT-based attack on sequential logic encryption. Section~\ref{sec:method} discusses the notion of functional corruptibility and its implications for SAT-based attacks.
Section~\ref{sec:attack} details the attack flow and implementation. 
In Section~\ref{sec:experiment}, we validate the effectiveness of the proposed attack in comparison with the previous SAT-based attack. Finally, we conclude the paper in Section~\ref{sec:conclude}.

\section{Preliminaries}\label{sec:background}

We first provide an introduction to sequential logic encryption and two state-of-the-art methods in this category. Then, we illustrate the mechanism of existing SAT-based attacks on sequential logic encryption. 

\subsection{Sequential Logic Encryption}

Sequential logic encryption methods encrypt the finite state machine (FSM) of a circuit via additional states and transitions~\cite{chakraborty2009harpoon,desai2013interlocking,meade2017revisit,dofe2018novel}. 
The FSM can operate in the \emph{encrypted mode} or in the \emph{functional mode}. 
After reset, the encrypted circuit is in the encrypted mode. It will transition to the functional mode, exhibiting the correct functionality, once it is provided with a correct sequence of multi-bit inputs, i.e., a correct key sequence, via the primary input ports. In this paper, we focus on two representative state-of-the-art sequential encryption methods, namely, \emph{HARPOON}~\cite{chakraborty2009harpoon} and \emph{Interlocking}~\cite{desai2013interlocking}, but describe applications to other decryption mechanisms in Section~\ref{sec:experiment}.

In \emph{HARPOON}~\cite{chakraborty2009harpoon}, only a single path is designed from the encrypted to the functional mode, which makes the scheme potentially vulnerable to attacks that can analyze the state transition diagram of the encrypted FSM to recognize the single transition from the encrypted to the functional mode~\cite{meade2017revisit}.  \emph{Interlocking}~\cite{desai2013interlocking} addresses this vulnerability by designing multiple paths between the encrypted and the functional mode so that the boundary between the logic in the two modes is less distinguishable. 
The circuit operates correctly only if the correct path is taken by  providing the associated key sequence. 
Otherwise, errors will still occur despite the circuit enters the functional mode.

\subsection{SAT-Based Attacks}\label{sec:sat_mechanism}

\begin{figure}[t]
\centerline{\includegraphics[width=0.9\columnwidth]{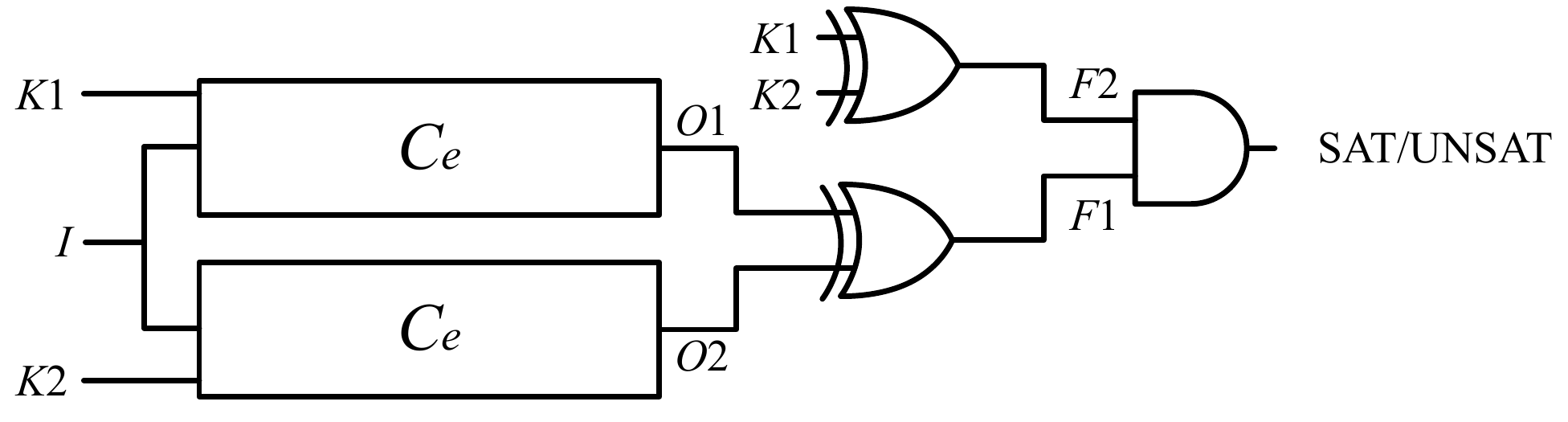}}
\caption{Circuit visualization of the SAT instance at the first iteration of the SAT attack.}
\vspace{-5mm}
\label{fig:miter}
\end{figure}

\begin{figure*}[t]
\centering
\subfigure[]{
\includegraphics[width=0.5\columnwidth]{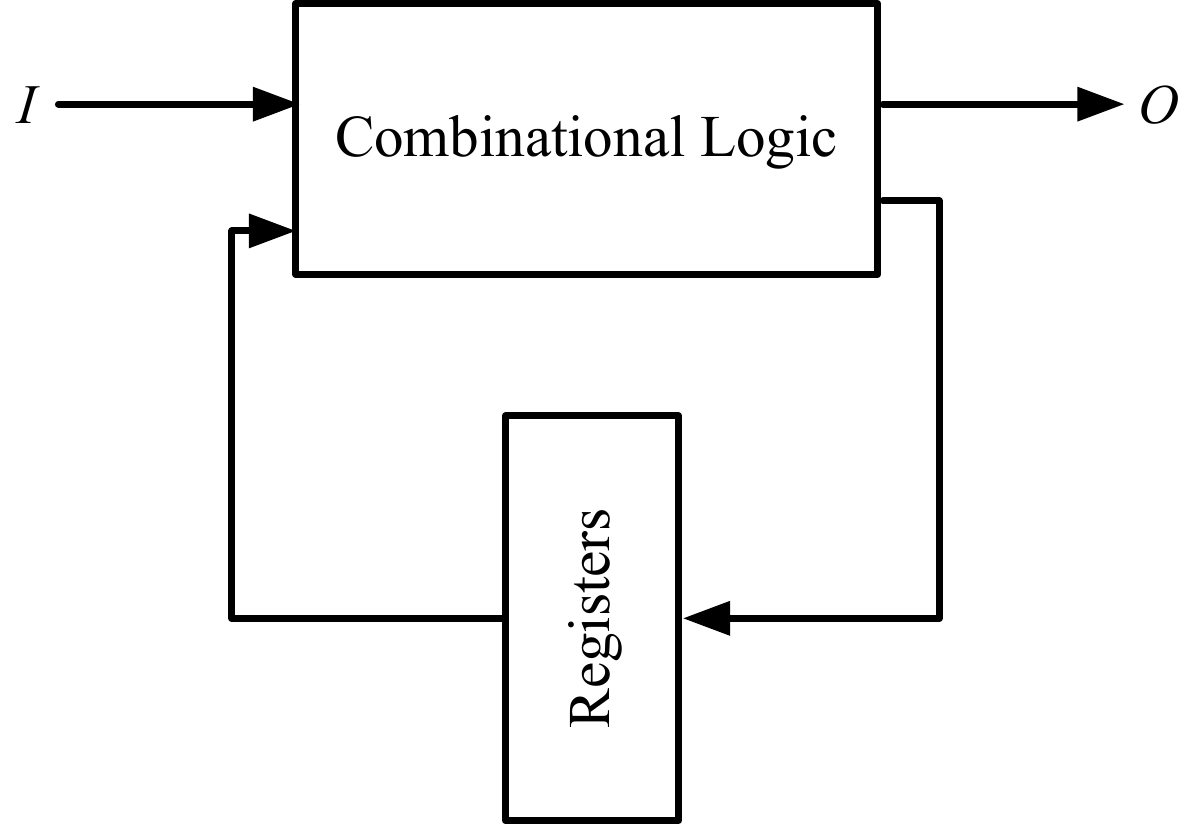}}
\subfigure[]{
\includegraphics[width=1.4\columnwidth]{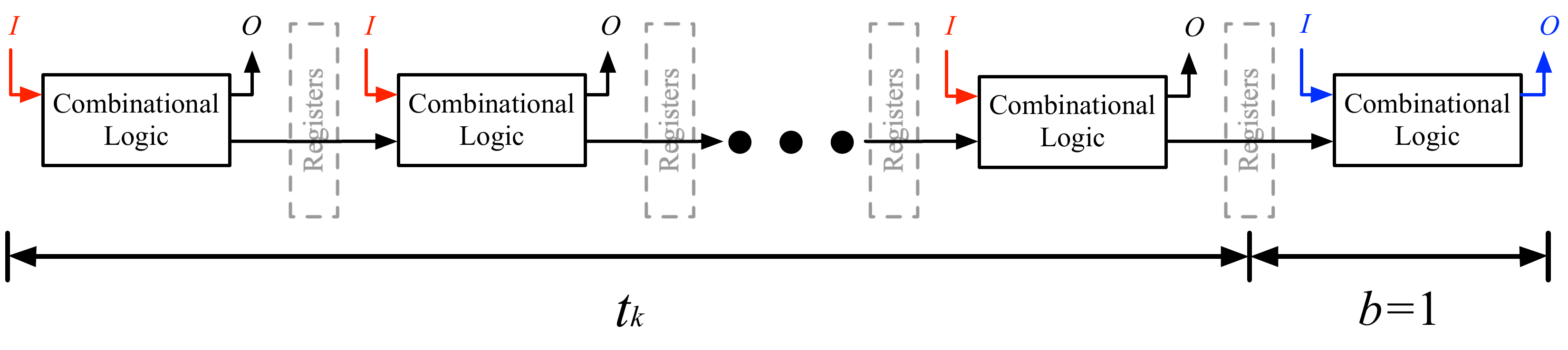}\label{fig:after_unroll}
}
\caption{Schematic of (a) a sequential circuit and (b) its unrolled version.}
\vspace{-4mm}
\label{fig:unroll}
\end{figure*}

The traditional SAT attack~\cite{subramanyan2015evaluating} as well as other SAT-based attacks assume the availability of two resources: the encrypted circuit netlist $C_{e}$ and an oracle circuit $C_{o}$, i.e., a black box providing the correct input/output response. 
The traditional SAT attack~\cite{subramanyan2015evaluating} can find the correct key of a combinational circuit which has no internal loops. The attack first constructs a SAT instance based on the circuit in Fig.~\ref{fig:miter}, where $I$ denotes the input ports, $K1$ and $K2$ the key ports, and $O1$ and $O2$ the output ports. 
Once a SAT instance is solved, the satisfying (SAT) assignment  identifies an input $i$ for which two different keys $k_1$ and $k_2$ lead to two different outputs. This input, called a \emph{distinguishing input pattern} (DIP), will be used to query the oracle $C_{o}$ for the corresponding correct output. 
The input/output pair from the oracle is then encoded into a set of clauses that are appended to the current SAT instance. Once a DIP $i^{dip}$ is found, the updated SAT instance effectively prunes out of the search space a set of wrong keys $\overline{K}$ as follows, 
$$\overline{K}=\{k|f'(i^{dip},k)\neq f(i^{dip})\},$$
where $f$ and $f'$ are the functions implemented by $C_{o}$ and $C_{e}$, respectively. 
The attack terminates when the updated SAT instance becomes unsatisfiable, meaning that all the wrong keys have been excluded, and any key in the remaining set can be returned as correct.

Because the SAT attack assumes that the internal state is part of the input/output response, it cannot be directly applied to sequential circuits when the internal state is not scannable.
A set of SAT-based attacks~\cite{shamsi2019kc2, el2017reverse} have circumvented this limitation by leveraging circuit unrolling and model checking. Similarly, the SAT attack can be extended to sequential logic encryption by unrolling a sequential circuit, as shown in Fig.~\ref{fig:unroll}, to form a larger combinational circuit that represents the behavior of the sequential circuit over a fixed number of clock cycles~\cite{meade2017revisit,hu2020sanscrypt,hu2020sanscrypt_extended}. We assume that the attacker knows the length of the key sequence $t_k$. The minimum number of unrollings required to start the attack is then $t_k+1$. The input ports in red and blue are treated as key ports and input ports, respectively, for the unrolled circuit. The output ports in blue act as the circuit output.

\begin{algorithm}[t]
 \caption{Reference SAT-Based Attack~\cite{shamsi2019kc2,el2017reverse,meade2017revisit}}
 \begin{algorithmic}[1]\label{alg:reference}
 \renewcommand{\algorithmicrequire}{\textbf{Input:}}
 \renewcommand{\algorithmicensure}{\textbf{Output:}}
\REQUIRE Encrypted netlist $C_{e}$, oracle $C_{o}$, key sequence length $t_k$
\ENSURE  Correct key sequence $k^*$
\STATE $b = 1$
\WHILE{\textbf{True}}
    \STATE $k^*, L_{dip}, L_{odip}=sat\_attack(C_{e}, C_{o}, t_k, b)$
    \IF{$!key\_verify(k^*, L_{dip}, L_{odip})$}
        \STATE $b = update(b)$
    \ELSE 
        \STATE break
    \ENDIF
\ENDWHILE
 \RETURN $k^*$
 \end{algorithmic} 
  \end{algorithm}

Algorithm~\ref{alg:reference} summarizes the flow of this attack. 
The combinational circuit in Fig.~\ref{fig:unroll} allows performing a SAT attack (line 3) to search for DIPs\footnote{For sequential circuits, a DIP can also be called a \emph{distinguishing input sequence} (DIS)~\cite{shamsi2019kc2}.} that can prune out wrong keys.
However, once a SAT attack terminates successfully for $t_k+b$ unrollings, there is no guarantee that the set of candidate keys, obtained by matching the oracle response over $b$ cycles, will also match its response after $b$ cycles.
A model checking problem~\cite{el2017reverse} can then be formulated (line 4) to verify whether this is the case, by taking as input one of the candidate keys $k^*$, the list of DIPs $L_{dip}$ and their corresponding outputs $L_{odip}$ generated by the SAT attack in line 3. Otherwise,  
the attack will try a different number of unrollings determined by an update function
(line 5) and repeat this process until all the remaining keys are proven correct. 
The update function usually increments or multiplies $b$ by a constant number. 
However, without any guidance on the unrolling depth that is required to eliminate all the wrong keys, 
Algorithm~\ref{alg:reference} tends to still require many SAT and model checking problems to be solved, especially when the required unrolling depth is large. 
\texttt{Fun-SAT} aims to significantly improve the efficiency of this attack on sequential logic encryption by directly estimating the number of unrollings that are required to prune out all the wrong keys. 

\section{Bounded-Depth Functional Corruptibility}\label{sec:method}

We introduce a notion of functional corruptibility (FC) for sequential circuits and discuss how it will be used to estimate the minimum number of required circuit unrollings for a successful SAT-based attack.  
Consistently with the literature, 
we assume the attacker's access to
the encrypted netlist $C_{e}$, the black-box oracle $C_{o}$, and the key length $t_k$. 
We discuss extensions to the case of unknown $t_k$ in Section~\ref{sec:experiment}. In the following, we say that the \emph{unrolling depth} is $b$ for $C_{o}$ and $C_{e}$ to mean that the circuit is unrolled for $b$ and $t_k+b$ cycles, respectively. 

\subsection{FC and Unrolling Depth}
\label{sec:FC_insight}

We denote by $b_{req}$ 
the minimum number of unrollings required to prune out all the wrong keys for a sequential SAT-based attack.
Let $I$ and $O$ be the sets of input and output ports of both $C_{o}$ and $C_{e}$, respectively. Let $f_b: \mathbb{B}^{b|I|}\rightarrow \mathbb{B}^{b|O|}$ be the function implemented by the $b$-unrolled version of $C_{o}$, i.e., the function represented by the combinational circuit $C^b_{o}$ obtained after unrolling $C_{o}$ for $b$ cycles. We also say that $C^b_{o}$ has depth $b$. Similarly, we denote by $f_b':\mathbb{B}^{b|I|}\times \mathbb{B}^{t_k|I|}\rightarrow \mathbb{B}^{b|O|}$ the $b$-unrolled version $C^b_{e}$ of $C_{e}$. 
We also denote by $SAT(b)$ the traditional SAT attack on $C^b_{e}$  and by $\overline{K}_{b}$ the set of wrong keys pruned out by $SAT(b)$. Finally, we denote by $|p|$ the length of a sequence $p$ and recall that a partial order can be defined over sequences as follows.
\begin{definition}[Partial Order Over Sequences]\label{def:domination}
Let $p$ and $q$ be two sequences with $|p|<|q|$. We say that $p$ (strictly) precedes (or is less than) $q$, written $p \prec q$, if and only if the following holds: $$ p_i=q_i, \forall i \in \{1, 2, ..., |p|\},$$
where $p_i$ is the $i$-th element of $p$, that is, if and only if $p$ is a prefix of $q$. Otherwise, we say that $p$ does not precede $q$, i.e., $p\nprec q$. 
\end{definition}
For example, $0110 \prec 011011$ holds while we have $0010 \nprec 011011$.
We introduce a notion of functional corruptibility for a sequential circuit by resorting to its $b$-unrolled version as follows:
\begin{definition}[$b$-Depth Functional Corruptibility]\label{def:fc_seq}
The $b$-depth functional corruptibility of a circuit pair $(C_{e}, C_{o})$ is the ratio between the number of corrupted output values of $C^b_{e}$ with respect to $C^b_{o}$ and the total number of primary input
and key combinations for $C^b_{e}$, i.e., 
$$FC_b = \frac{1}{2^{(b+t_k)|I|}} \sum_{i\in \mathbb{B}^{b|I|}} \sum_{k\in \mathbb{B}^{t_k|I|}}\mathbbm{1}(f_b(i)\neq f_b'(i,k)),$$
where $\mathbbm{1}(.)$ is the indicator function. 
\end{definition}

To study how $FC_b$ evolves with $b$, we associate a tag to the errors introduced by $f_b'$, the function implemented by the $b$-unrolled version $C^b_{e}$ of the encrypted circuit.

\begin{definition}[Error Tag]\label{def:inherit}
We can associate a tag to an input $i\in \mathbb{B}^{b|I|}$ and a key $k\in \mathbb{B}^{t_k|I|}$
via the map $\mathcal{T}_b: \mathbb{B}^{b|I|} \times \mathbb{B}^{t_k|I|} \to \mathbb{N}\cup\{\bot\}$ defined as follows: 
\begin{equation}\label{opt_formulation}\notag
    \begin{aligned}
    &\mathcal{T}_b(i,k) = \left\{
        \begin{aligned}
            & \bot,\ {\rm if}\ f_{b}'(i,k)= f_{b}(i);\\
            & 1,\ {\rm if}\ f_{b}'(i,k)\neq f_{b}(i)\ {\rm and }\ b=1;\\
            & b,\ {\rm if}\ f_{b}'(i,k)\neq f_{b}(i),\ b>1,\\ 
            &\ \ \ \ \ \ {\rm and }\ f_{b-1}'(j,k)= f_{b-1}(j);\\
            & \mathcal{T}_{b-1}(j,k),\ {\rm otherwise;}
         \end{aligned}
     \right.
    \end{aligned}
\end{equation}
where $j \prec i$ and $j\in \mathbb{B}^{(b-1)|I|}$.
\end{definition}

\begin{figure}[t]
\centering
\includegraphics[width=0.8\columnwidth]{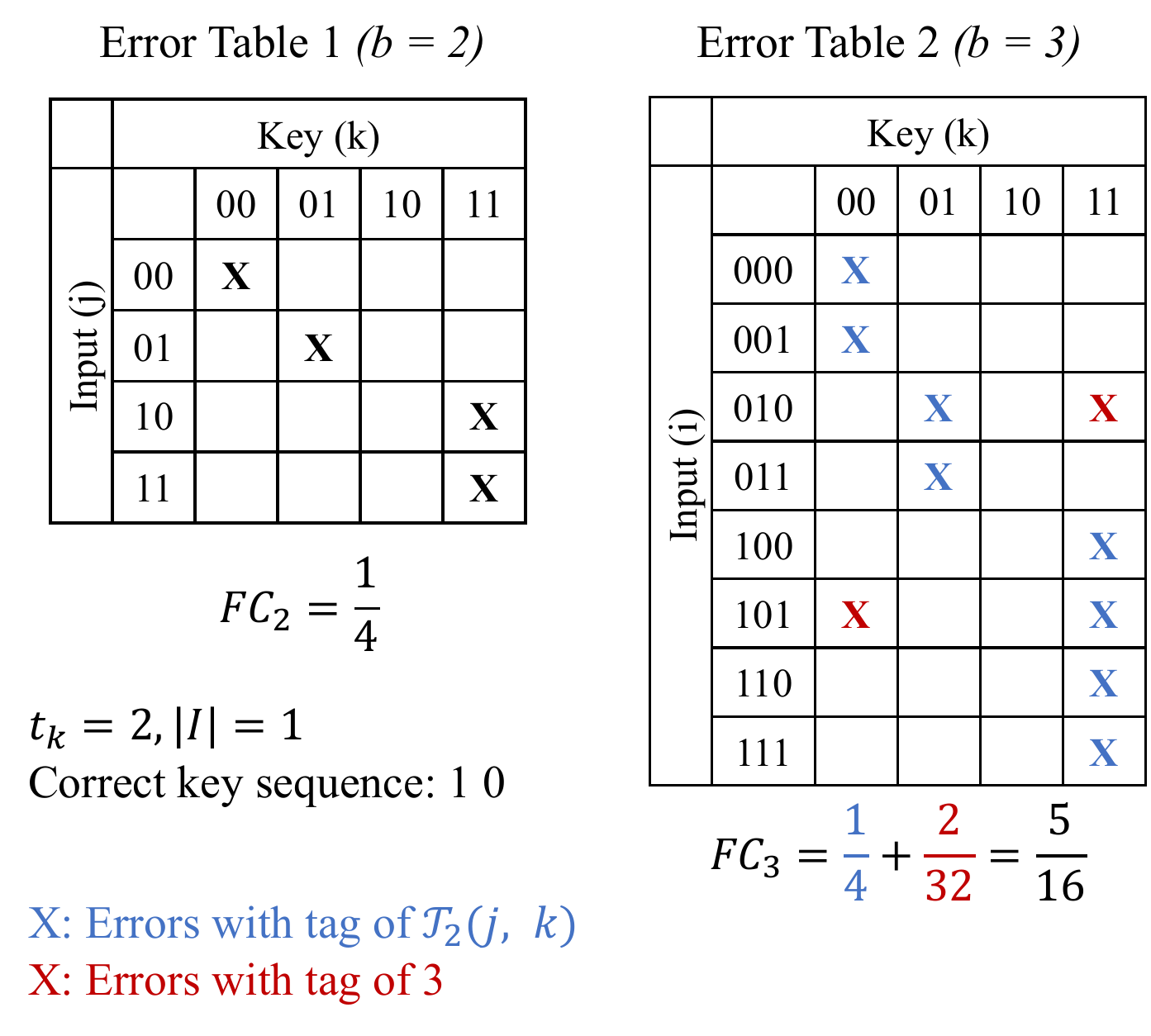}
\vspace{-3mm}
\caption{Error distributions of a small encrypted circuit ($|I|=1$ and $t_k=2$) for the first two and three clock cycles.}
\vspace{-5mm}
\label{fig:error_type}
\end{figure}

For example, Fig.~\ref{fig:error_type} shows two tables marking the errors introduced by $f_{b}'$ for $b=2$ and $b=3$, with $|I|=1$ and $t_k=2$. Each entry is indexed by an input value $i$ and a key value $k$. If $f_{b}'(i,k)\neq f_{b}(i)$, we mark the corresponding entry with an ``x.'' By Definition~\ref{def:inherit}, the errors marked in red are tagged with $3$, while the blue ones depend on errors that were already introduced by $f_{2}'$ and will be tagged based on the error tags they had in $f_{2}'$. We can then state the first result describing the behavior of $FC_b$.

\begin{theorem}\label{theorem:fc_eq}
For all $b>1$, $FC_b=FC_{b-1}$ holds if and only if $\mathcal{T}(i,k) < b$ holds $\forall \ i\in \mathbb{B}^{b|I|}$, $\forall \ k\in \mathbb{B}^{t_k|I|}$, i.e., if and only if no new errors, i.e., errors tagged by $b$, are introduced by $f_b'$. 
\end{theorem}
\begin{proof}

We first prove that $FC_b=FC_{b-1}$ holds if $f_b'$ does not introduce new errors, i.e., errors tagged by $b$. Suppose that $n_{b-1}$ is the number of errors introduced by $f_{b-1}'$. 
Assume that an input $j\in \mathbb{B}^{(b-1)|I|}$ and a key $k\in \mathbb{B}^{t_k|I|}$ lead to an output error, i.e., 
$f_{b-1}(j)\neq f_{b-1}'(j,k)$ holds. By Definition~\ref{def:domination}, there are $2^{|I|}$ inputs $i$ such that $i \in \mathbb{B}^{b|I|}$ and $j\prec i$ holds. Therefore, $f_b'$ inherits $n_{b-1} 2^{|I|}$ errors from $f_{b-1}'$.  
By Definition~\ref{def:fc_seq}, we obtain
$$FC_b=\frac{n_{b-1}\cdot 2^{|I|}}{2^{b|I|+t_k|I|}}=\frac{n_{b-1}}{2^{(b-1)|I|+t_k|I|}}=FC_{b-1}.$$

We now prove that $f_b'$ does not introduce new errors if $FC_b=FC_{b-1}$ holds. Suppose by contradiction that $FC_b=FC_{b-1}$ holds and $f_b'$ introduces indeed new errors. Let 
$n^{<b}$ and $n^{=b}$ be the number of errors with tags less than and equal to $b$ in $f_b'$, respectively. By Definition~\ref{def:inherit}, the number of errors in $f_{b-1}'$ is 
$$n= \frac{n^{<b}}{2^{|I|}}.$$
Therefore, we have 
$$FC_{b-1} = \frac{n}{2^{(b-1)|I|+t_k|I|}}= \frac{n^{<b}}{2^{b|I|+t_k|I|}},$$
and 
$$FC_b =\frac{n^{<b}+n^{=b}}{2^{b|I|+t_k|I|}},$$
leading to $FC_{b-1} < FC_b$, which violates our initial assumption. 
\end{proof}

From Theorem~\ref{theorem:fc_eq} and its proof we infer that $FC_b$ can only remain constant or increase with $b$ as new errors, tagged by $b$, are observed at the output of the $b$-unrolled circuit. This leads to the following result, stating that $FC_b$ monotonically increases with $b$, independently of the sequential encryption method adopted.  

\begin{corollary}
\label{theorem:monotonic}
For all $b>1$, we obtain $FC_{b} \geq FC_{b-1}$. 
\end{corollary}

\subsection{FC and Key Search Progress}\label{sec:implications}

The $b$-depth functional corruptibility can be related to the set $\overline{K}_b$ of incorrect keys that are excluded by $SAT(b)$. 
To establish this relation, we first provide a characterization of  $\overline{K}_b$ using the following lemma. 

\begin{lemma}\label{lem:dip}
Let $\overline{K}_b$ be the set of wrong keys pruned out of the search space upon termination of $SAT(b)$. Then, $ \forall \ i\in\mathbb{B}^{b|I|}$, the following holds:  
$$\{k|f_b(i)\neq f_b'(i,k)\}\subseteq \overline{K}_b.$$
\end{lemma}
\begin{proof}
We articulate the proof into two cases based on whether the input $i$ is explicitly found as a DIP by $SAT(b)$ (case 1) or not (case 2). 

\emph{Case 1.} By the definition in Section~\ref{sec:sat_mechanism}, $\{k|f_b(i)\neq f_b'(i,k)\}$ is the set of keys that are pruned out by the DIP $i$. Therefore, we have $\{k|f_b(i)\neq f_b'(i,k)\}\subseteq \overline{K}_b$. 

\emph{Case 2.} 
We denote by $k^*$ one of the correct keys. 
Suppose, by contradiction, that
there exists a wrong key $k'$ that does not belong to $\overline{K}_b$ while satisfying $f_b(i)\neq f_b'(i,k')$. From the circuit visualization of the SAT instance 
in Fig.~\ref{fig:miter}, if $K1$, $K2$, and $I$ are assigned with $k^*$, $k'$, and $i$, respectively, then the SAT instance is satisfied. By definition, this means that $i$ is a DIP that would be found by $SAT(b)$, which contradicts our initial assumption. Therefore, there is no such wrong key $k'$ that does not belong to $\overline{K}_b$, i.e., $\{k|f_b(i)\neq f_b'(i,k)\}\subseteq \overline{K}_b$ holds. 
\end{proof}

Lemma~\ref{lem:dip} shows that a key $k'$ is in $\overline{K}_b$ if $f_b(i)\neq f_b'(i,k')$ holds, even if $i$ is not explicitly found as a DIP by $SAT(b)$. 
The following theorems characterize the relation between $FC_b$ and $\overline{K}_{b}$, stating that  $\overline{K}_{b}$ monotonically increases with $FC_b$ and remains the same if $FC_b$ remains constant.   

\begin{theorem}\label{theorem:dip} 
For all $b>1$, let $FC_b = FC_{b-1}$ hold. 
Let $\overline{K}_{b-1}$ and $\overline{K}_{b}$ be the sets of wrong keys pruned out by $SAT(b-1)$ and $SAT(b)$, respectively. Then, $\overline{K}_{b-1}=\overline{K}_{b}$ holds. 
\end{theorem}

\begin{proof}
We first prove that $\overline{K}_{b-1}\supseteq \overline{K}_{b}$ holds. Let $i_{b}^{dip}$ be a DIP found by $SAT(b)$. 
The set of wrong keys that can be pruned out by this DIP is  
$$\overline{K}=\{k|f_b'(i_b^{dip},k)\neq f_b(i_b^{dip})\}.$$
Since $FC_b = FC_{b-1}$, the tags of all errors introduced by $f_b'$ are less than $b$
according to Theorem~\ref{theorem:fc_eq}. Therefore, by Definition~\ref{def:inherit}, there exists an input $i_{b-1}\in \mathbb{B}^{(b-1)|I|}$ such that $i_{b-1}\prec i_{b}^{dip}$ and the following condition holds: 
$$f_{b-1}'(i_{b-1},k)\neq f_{b-1}(i_{b-1}),\ \forall k\in \overline{K}.$$
By Lemma~\ref{lem:dip}, 
$SAT(b-1)$ can prune out the wrong keys in $\overline{K}$. 
Therefore, $SAT(b-1)$ can prune out all the wrong keys pruned by $SAT(b)$, hence $\overline{K}_{b-1}\supseteq \overline{K}_{b}$ holds. 

We now prove that $\overline{K}_{b-1}\subseteq \overline{K}_{b}$ holds. Suppose $i_{b-1}^{dip}$ is a DIP found by $SAT(b-1)$. The set of wrong keys that can be pruned out by this DIP is 
$$\overline{K} = \{k|f_{b-1}'(i_{b-1}^{dip},k)\neq f_{b-1}(i_{b-1}^{dip})\}.$$
For any input $i_b$, such that $i_{b-1}^{dip}\prec i_b$, we have 
$$f_{b}'(i_b,k)\neq f_{b}(i_b),\ \forall k\in \overline{K}.$$
By Lemma~\ref{lem:dip}, 
$SAT(b)$ can prune out the wrong keys in $\overline{K}$. 
Therefore, $SAT(b)$ can prune out all the wrong keys pruned by $SAT(b-1)$, i.e., $\overline{K}_{b-1}\subseteq \overline{K}_{b}$ also holds. 
\end{proof}

\begin{theorem}\label{theorem:better_dip}
For all $b>1$, let $FC_b > FC_{b-1}$ hold, and let $\overline{K}_{b-1}$ and $\overline{K}_{b}$ be the sets of wrong keys pruned out by $SAT(b-1)$ and $SAT(b)$, respectively. Then, $\overline{K}_{b-1}\subseteq \nobreak \overline{K}_{b}$ holds. 
\end{theorem}

\begin{proof}
Since $FC_b > FC_{b-1}$ holds, $f_b'$ will also introduce new errors, with tags equal to $b$. 
For an arbitrary wrong key $k\in \overline{K}_{b}$, we consider the following two cases based on whether the following condition holds (case 1) or does not hold (case 2):
\begin{equation} \label{eq:split}
\exists \ i\in \mathbb{B}^{b|I|}: (f_b'(i,k)\neq f_b(i)) \wedge (\mathcal{T}_b(i,k)<b),
\end{equation}
i.e., 
whether there exists at least one error introduced by $f_b'$ whose tag is less than $b$.

\emph{Case 1.}
By Theorem~\ref{theorem:dip}, there exists an input $i_{b-1}\in \mathbb{B}^{(b-1)|I|}$ such that
$f_{b-1}'(i_{b-1},k)\neq f_{b-1}(i_{b-1})$ holds. Therefore, by Lemma~\ref{lem:dip}, the wrong key $k$ is also in $\overline{K}_{b-1}$. 

\emph{Case 2.} 
Since $k\in \overline{K}_{b}$, the associated errors must all have tag $b$, i.e., 
$$ \mathcal{T}_b(i,k)=b, \forall i: f_b'(i,k)\neq f_b(i).$$
By 
Definition~\ref{def:inherit}, we also conclude:
$$f_{b-1}'(i_{b-1},k)= f_{b-1}(i_{b-1}),\ \forall \ i_{b-1}\in \mathbb{B}^{(b-1)|I|},$$
i.e., $k$ is not in $\overline{K}_{b-1}$.  

Based on whether all the wrong keys in $\overline{K}_{b}$ satisfy~\eqref{eq:split} or not, we conclude that $\overline{K}_{b}$ can be at least equal to and possibly a superset of $\overline{K}_{b-1}$. Therefore,  
$\overline{K}_{b-1}\subseteq \overline{K}_{b}$ holds. 
\end{proof}

Theorem~\ref{theorem:dip} and Theorem~\ref{theorem:better_dip} directly relate the behavior of  $\overline{K}_{b}$, hence the progress made by the attack, to the behavior of $FC_{b}$, which can be efficiently approximated via logic simulation. Therefore, by efficiently estimating the sequence $FC_b$, we can look ahead and make informed predictions about whether a SAT-attack instance or a model checking instance need to be solved or can be skipped to rapidly progress with the attack. 

Specifically, if we find that $FC_b = FC_{b-1}$ holds, then we infer that it is sufficient to execute $SAT(b-1)$, which is a smaller and usually faster instance than $SAT(b)$. In fact, we know from Theorem~\ref{theorem:dip} that $SAT(b)$ cannot exclude any more wrong keys other than those already pruned out by $SAT(b-1)$. On the other hand, 
when $FC_b > FC_{b-1}$ holds, we infer that $SAT(b)$ is at least as effective as $SAT(b-1)$, or even more effective,  in narrowing down the search for the correct key. We can then move forward and directly execute $SAT(b)$ by skipping the execution of $SAT(b-1)$ and the additional model checking problem needed to verify termination on $C^{b-1}_{e}$. We incorporate both of these insights into \texttt{Fun-SAT}, as detailed below.    

\section{Functional Corruptibility-Guided SAT-Based Attack}\label{sec:attack}

We detail the attack flow of \texttt{Fun-SAT} and discuss its termination conditions.

\subsection{Attack Flow}

\begin{algorithm}[t]
 \caption{\texttt{Fun-SAT}}
 \begin{algorithmic}[1]\label{alg:attack_main}
 \renewcommand{\algorithmicrequire}{\textbf{Input:}}
 \renewcommand{\algorithmicensure}{\textbf{Output:}}
\REQUIRE Encrypted netlist $C_{e}$, oracle $C_{o}$, key sequence length $t_k$, FC analysis window $t_{win}$, FC difference threshold $\delta$, FC hold threshold $\Delta$, simulation sample size $S$
\ENSURE  Correct key sequence $k^*$

\STATE $b_l = 1;\ b_u=t_{win}$
\WHILE{\textbf{True}}
    \STATE // ****** FC analysis phase ******
    \STATE $counter = 0$
    \FOR{$b=b_l$ to $b_u$}
        \STATE $FC_b =simulate(C_{e}, C_{o}, t_k, b, S)$
        \IF{$b > 1$}
            \IF{$FC_b-FC_{b-1}\leq \delta$}
                \STATE $counter = counter + 1$
            \ELSE
                \STATE $counter = 0$
            \ENDIF
        \ENDIF
        \STATE $b^*=b$
        \IF{$counter == \Delta$}
            \STATE $b^*=b^*-\Delta$
            \STATE break
        \ENDIF
    \ENDFOR
    \STATE // ****** SAT attack phase ******
    \STATE $k^*, L_{dip}, L_{odip}=sat\_attack(C_{e}, C_{o}, t_k, b^*)$
    \IF{$!key\_verify(k^*, L_{dip}, L_{odip})$}
        \STATE $b_l = b^* + 1;\ b_u = b^* + 1 + t_{win}$
    \ELSE 
        \STATE break
    \ENDIF
\ENDWHILE
 \RETURN $k^*$
 \end{algorithmic} 
 \end{algorithm}

As shown in Algorithm~\ref{alg:attack_main}, \texttt{Fun-SAT} accepts as inputs a set of circuit-related and attack configuration parameters. Circuit-related inputs consist of the encrypted netlist $C_{e}$, the oracle $C_{o}$, and the key sequence length $t_k$. Configuration parameters include the FC analysis window $t_{win}$, the FC difference threshold $\delta$, the FC hold threshold $\Delta$, and the simulation sample size $S$, further described below. 

The attack consists of an FC analysis phase (line 3 to line 19) and a SAT attack phase (line 20 to line 26). 
The first phase analyzes a sequence of FC values for different unrolling depths $b$ 
and predicts the number of unrollings to be used in the second phase. 
To start the first phase, we specify the initial range of $b$, i.e., the number of unrollings excluding the key length $t_k$, to be $[1, t_{win}]$ in line 1.  
$FC_b$ is obtained via a logic simulation function in line 6, which simulates both $C_{e}$ and $C_{o}$ with random inputs and keys for $S$ times and calculates an estimate of $FC_b$ based on Definition~\ref{def:fc_seq}. 

\begin{figure}[t]
    \centering
    \includegraphics[width=0.95\columnwidth]{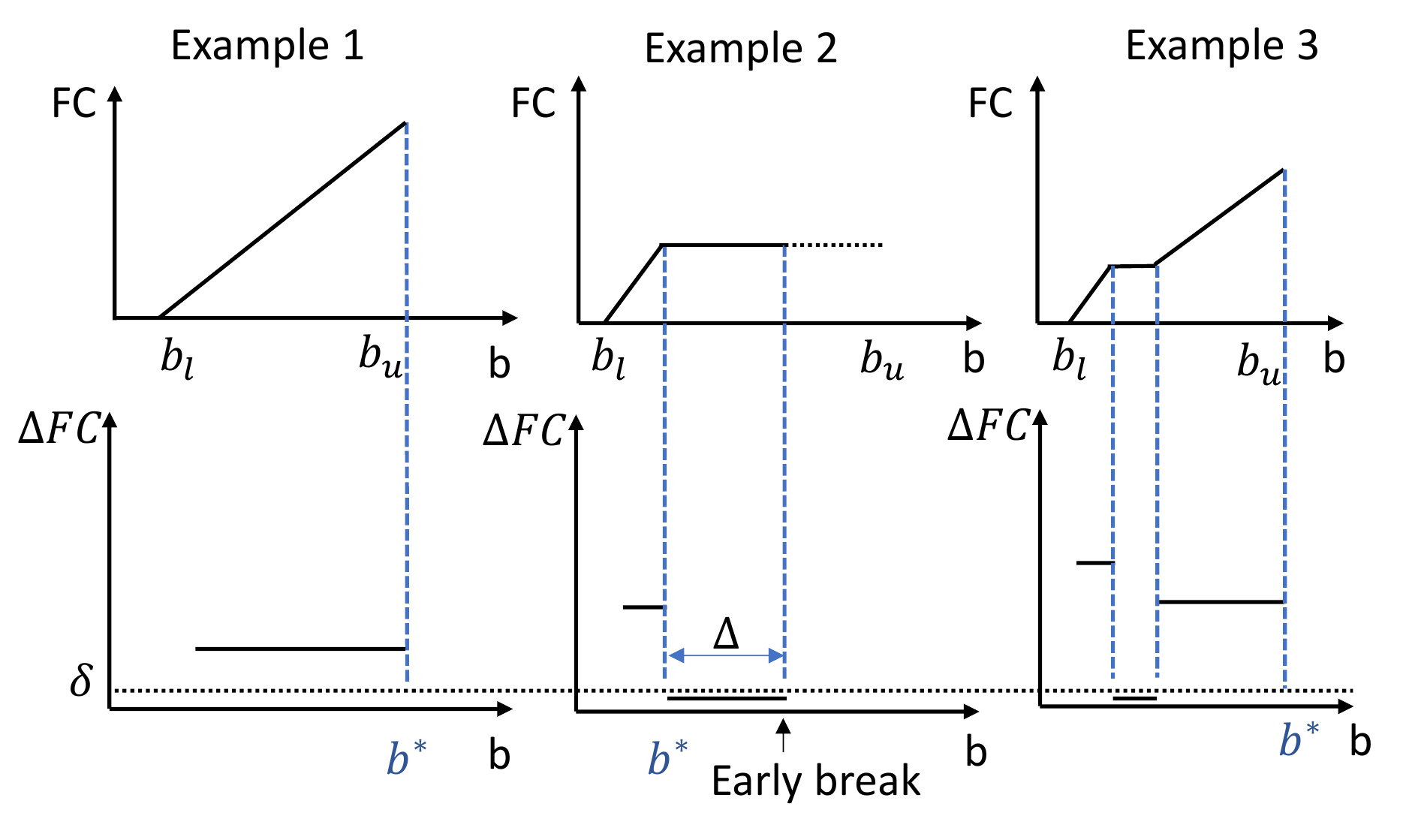}
    \vspace{-5mm}
    \caption{Sample behaviors of $FC$ as a function of $b$ ($\Delta FC_b = FC_b - FC_{b-1}$).}
    \vspace{-5mm}
    \label{fig:FC_2}
\end{figure}

As discussed in Section~\ref{sec:implications}, if $FC_b > FC_{b-1}$ holds for some $b$, then we pick $b$ as the number of unrollings for the second phase (based on Theorem~\ref{theorem:better_dip}). Otherwise, if $FC_b = FC_{b-1}$ holds, we use $b-1$ (based on Theorem~\ref{theorem:dip}). The FC analysis step (line 7 to line 18) implements this decision rule as follows. 
If $FC$ keeps increasing over the analysis window, we select the upper bound of the window $b_u$, as  shown in Example 1 in Fig.~\ref{fig:FC_2}.  
Otherwise, we relax the equality condition to accommodate approximation errors and check whether $(FC_b-FC_{b-1})\leq \delta$ holds, where $\delta$ is a small positive number (e.g., $0.01$). If this condition holds for very few occurrences before $FC$ increases again, as in Example~3 in Fig.~\ref{fig:FC_2}, 
we count these occurrences and terminate the first phase with an early break 
only when the number of successive occurrences exceeds a pre-determined threshold $\Delta$. 
This early break prevents unnecessary FC simulations when $FC$ stops increasing at an early stage, as in  Example~2 in Fig.~\ref{fig:FC_2}. 

In the second phase, the SAT attack is executed (line 21) with the number of unrollings $b^*$ decided in the first phase. However, 
the resulting candidate key $k^*$ can only guarantee the correct behavior of the sequential circuit up to $b^*$ clock cycles after reset. Therefore, an additional key verification step is needed to check whether $k^*$ is indeed the correct key (line 22). The candidate key $k^*$, the list of found DIPs $L_{dip}$ and the list of the corresponding correct output values 
$L_{odip}$ are used in this step, described in  Section~\ref{sec:termination_condition}.  
When the key verification fails, we revisit the first phase and search for a new unrolling depth value with a new FC analysis window 
(line 23). 
This loop terminates when a correct key is found.

\subsection{Termination Conditions}\label{sec:termination_condition}

The last step before the successful termination of the attack in Algorithm~\ref{alg:attack_main} is the key verification function (line 22). In this step, because the netlist of the oracle is assumed unavailable in the attack model,
we cannot directly check
the equivalence of the oracle $C_{o}$ and the encrypted netlist $C_{e}$ configured with $k^*$. 
We therefore build on the literature on SAT-based attacks~\cite{el2017reverse} to formulate  
two key verification conditions, namely, unique key (UK) and model checking equivalence (MCE),
and conclude on the correctness of a candidate key for a circuit encrypted via sequential logic encryption.

\fakeparagraphnospace{Unique Key} If $k^*$ is the only remaining key that makes the encrypted circuit $C_{e}$ behave the same as the oracle $C_{o}$ for all DIPs in $L_{dip}$, then $k^*$ is the correct key.
We check this condition by verifying that the set 
$$\{k'|f'_b(i,k')=f'_b(i,k^*),\forall i \in L_{dip} \}\backslash\{k^*\}$$
is empty. To do so, we construct 
a SAT instance as shown in Fig.~\ref{fig:uk_model}, where $C_{e}^b$ is the $b$-unrolled version of $C_{e}$ and $L_{dip}[n]$ is the $n$-th DIP in $L_{dip}$. When this SAT instance is UNSAT, the UK condition is satisfied and the attack terminates successfully with the correct key $k^*$. 

\begin{figure}[t]
\centerline{\includegraphics[width=0.9\columnwidth]{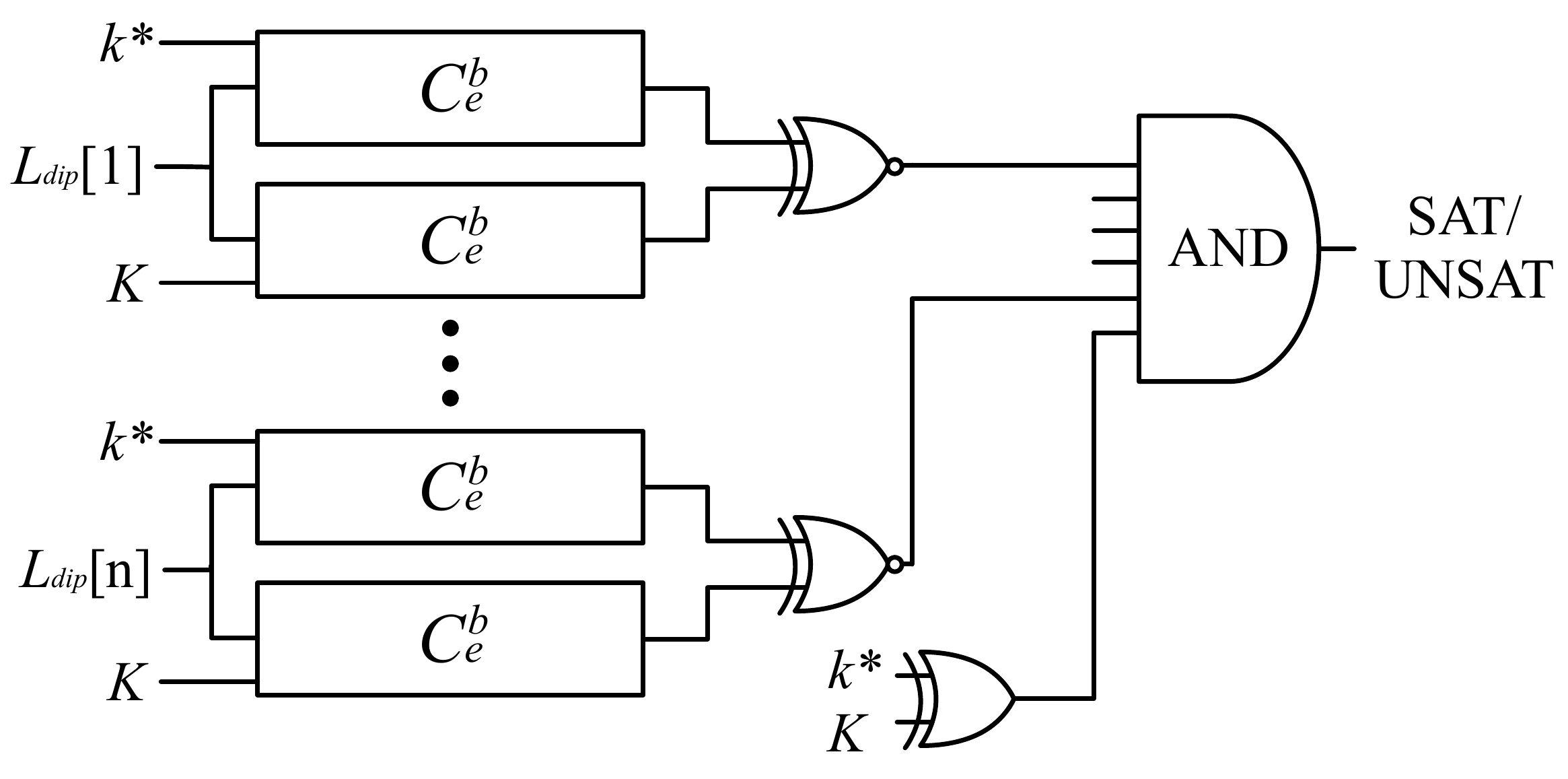}}
\caption{Circuit visualization of the SAT instance to check the unique key (UK) condition.}
\vspace{-5mm}
\label{fig:uk_model}
\end{figure}

\fakeparagraphnospace{Model Checking Equivalence} 
The SAT attack in Algorithm~\ref{alg:attack_main} often returns a set of candidate keys that match the behavior of the oracle over the first $b$ clock cycles, but may eventually lead to different outputs over longer horizons. 
We detect these spurious keys by leveraging model checking to search whether there exist two keys in the candidate set that lead to different outputs on the encrypted circuit $C_{e}$. 

\begin{figure*}[t]
\centerline{\includegraphics[width=0.7\textwidth]{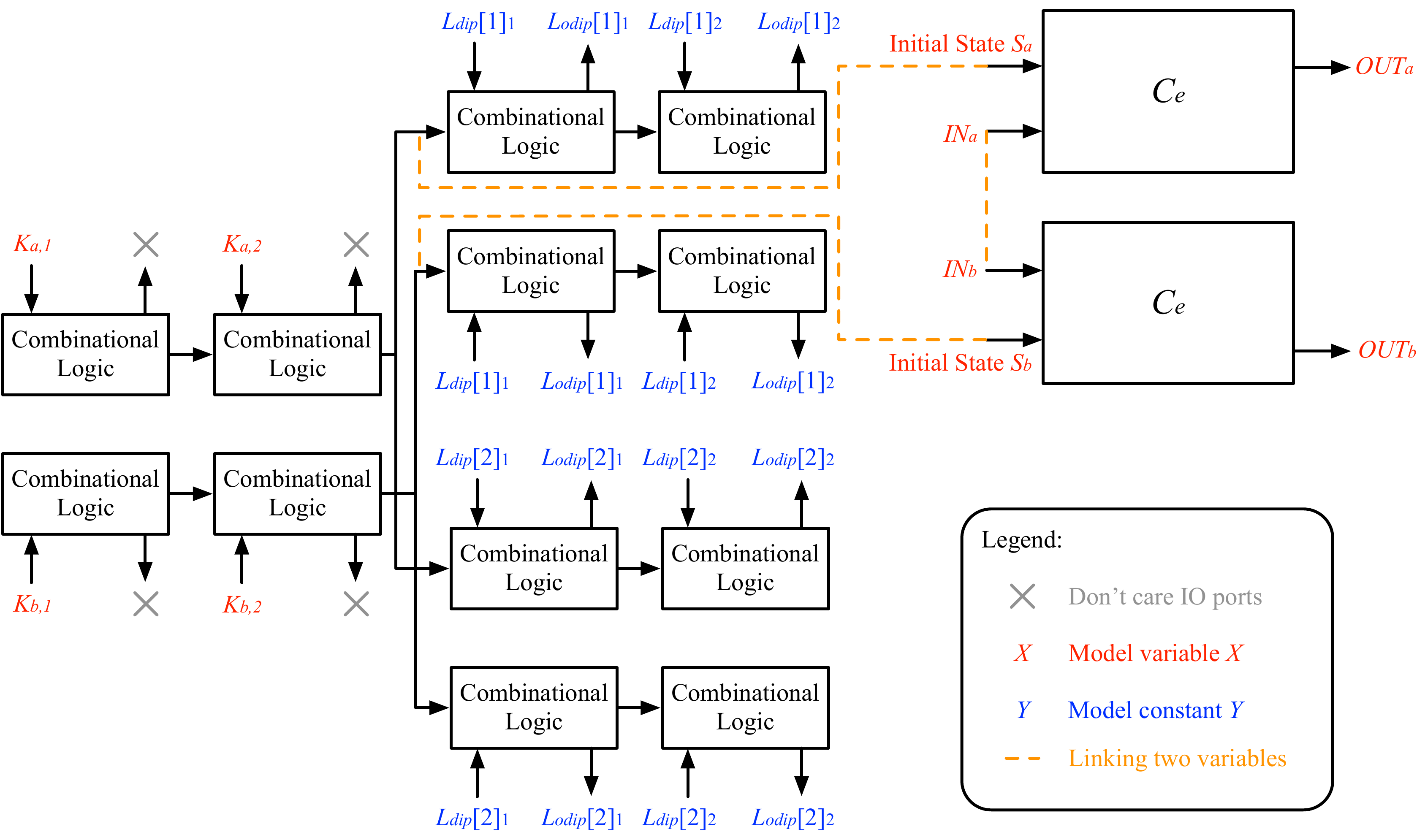}}
\caption{The visualization of the model construction for checking whether the remaining keys are all correct.}
\label{fig:mc_model}
\vspace{-5mm}
\end{figure*}

Fig.~\ref{fig:mc_model} shows the circuit visualization of an example model used for this task, where $t_k=b=|L_{dip}|=|L_{odip}|=2$. For better illustration purpose, we represent the unrolled circuit $C_{e}^2$ as a cascade of $t_k+b$, i.e., four copies of the combinational logic of $C_{e}$.  The IO ports in red, blue, or grey, are treated by the model checker as free binary variables, constant values, or `don't care' bits, respectively.
$K_a$ and $K_b$ are the two keys. $K_{a,1}$ and $K_{a,2}$ represent the portion of $K_a$ in the first and the second clock cycle, respectively. Similar subscript indices are used for $K_b$, the DIPs in $L_{dip}$, and the corresponding correct outputs in $L_{odip}$.  
The circuits on the left side of Fig.~\ref{fig:mc_model}, whose inputs and outputs are assigned all the DIPs and the corresponding correct outputs, are used to model the constraints that limit the search space for $K_a$ and $K_b$ only to the candidate set.  
On the right side of Fig.~\ref{fig:mc_model}, two copies of the encrypted circuit $C_{e}$ are instantiated and forced to receive the same inputs at $IN_a$ and $IN_b$, as denoted by the orange dashed lines linking the two ports. Importantly,  the initial states $S_a$ and $S_b$ of the two copies of $C_{e}$ are also set to the corresponding initial states obtained after applying $K_a$ or $K_b$ to the left-side circuit, respectively. 
These constraints ensure that the two copies of $C_{e}$ on the right side are configured as if they were provided with $K_a$ or $K_b$. Such a construction was not used in previous  work~\cite{el2017reverse,shamsi2019kc2}, since $K_a$ and $K_b$ could be directly provided to $C_{e}$ via additional key ports.   

The model above is given to a model checker to verify whether $OUT_a=OUT_b$ holds for 
an unbounded horizon. 
If the model checker returns false, it means there exists at least one key in the candidate set that eventually leads to the wrong circuit behavior on $C_{e}$. 
Therefore, the attack may not terminate. 
In our implementation, we first perform bounded model checking (BMC) with bound $b+1$ as a preliminary check. If BMC returns false, we can already conclude that the SAT-based attack should continue. Otherwise, we 
need to perform unbounded model checking to verify the correctness of the candidate key set beyond $b+1$ cycles. The outcome of unbounded model checking
will determine whether to conclude or continue the attack. 

\section{Experimental Results}\label{sec:experiment}

\fakeparagraphnospace{Experiment Setup}
\texttt{Fun-SAT} was implemented in Python and executed on a Linux server with 48 $2.1$-GHz cores and $500$-GB memory. 
In the FC analysis phase, we use Synopsys VCS to run the logic simulation on both the oracle and the encrypted netlist. 
The logic simulation time is deemed as a reasonable, if not conservative, approximation of the functional query time, since hardware execution is generally faster than software simulation.
The SAT attack is adapted from the literature~\cite{subramanyan2015evaluating} 
and leverages \textsc{MiniSAT}~\cite{sorensson2005minisat} as the SAT solver while model checking is performed using the interpolation-based  algorithm~\cite{mcmillan2003interpolation}  
implemented by the  \textsc{nuXmv}~\cite{cavada2014nuxmv}  model checker.  
The model is constructed as described in Section~\ref{sec:attack} and encoded into the \textsc{nuXmv} format by a Python script.
For the reference attack in Algorithm~\ref{alg:reference}, we first evaluate increasing the unrolling depth $b$ by one whenever the key verification fails. 
We then explore a multiplicative update rule, e.g., $b = 2\cdot b$, which can speed up both the reference and \texttt{Fun-SAT} attacks.
We select ten benchmarks from ISCAS'89~\cite{brglez1989combinational} and ITC'99~\cite{davidson1999itc}, as detailed in Table~\ref{tab:benchmark}, 
in line with the sizes of the benchmarks used in the related literature~\cite{shamsi2019kc2,el2017reverse,hu2020sanscrypt_extended}.

\begin{table}[t]
    \caption{Overview of the Selected Benchmark Circuits} 
    \vspace{-8mm}
    \begin{center}
    \resizebox{\columnwidth}{!}{
    \begin{tabular}{|c|c|c|c|c|c|c|c|c|c|c|}
    \hline
    \textbf{Circuit} & \textbf{s27} & \textbf{s526} & \textbf{s1488} & \textbf{s9234} & \textbf{s15850} & \textbf{s38584} & \textbf{b10} & \textbf{b12} & \textbf{b15} & \textbf{b19} \\\hline
    \textbf{Inputs} & 4 & 3 & 8 & 19 & 14 & 12 & 11 & 5 & 36 & 24 \\\hline
    \textbf{Outputs} & 1 & 6 & 19 & 22 & 87 & 278 & 6 & 6 & 70 & 30 \\\hline
    \textbf{DFFs} & 3 & 21 & 6 & 228 & 597 & 1452 & 17 & 121 & 449 & 6642 \\\hline
    \textbf{Gates} & 10 & 193 & 653 & 5597 & 9772 & 19253 & 214 & 1217 & 8169 & 190213 \\\hline
    \end{tabular}}
    \label{tab:benchmark}
    \end{center}
    \vspace{-8mm}
\end{table}

We implement the two sequential logic encryption methods, as described in the reference papers~\cite{chakraborty2009harpoon,desai2013interlocking}. 
Besides the key size $t_k$, there are other configuration parameters that are specific to the methods. 
In HARPOON, the Modification Kernel Function (MKF) module is inserted to corrupt the circuit function when the applied key is wrong. In our experiments, we randomly select the locations of 
the MKFs and use the ratio $R_{mkf}$ between the number of MKFs and the number of gates in the original circuit as a parameter to configure the encryption. 
In Interlocking, a set of wrong keys can also bring the circuit to 
the functional mode without triggering output errors immediately.  
These keys trigger output errors only when the circuit enters certain predetermined states in the functional mode. Based on the distance between one such state $s_{pre}$ and the true reset state in the functional mode, the occurrence of output errors may be delayed for a few cycles. 
In our experiments, we introduce the parameter $D_{max}$ to set 
the maximum number of state transitions between the true reset state in the functional mode and a state $s_{pre}$ in which output errors are triggered. A wrong key $k_w$ is randomly assigned to $s_{pre}$ such that an output error occurs in $s_{pre}$ when $k_w$ is applied. 

We determine the simulation sample size $S$ in Algorithm~\ref{alg:attack_main} using numerical experiments such as the one reported in Fig.~\ref{fig:pre_exp}, where the FC of the largest benchmark $s38584$, encrypted with a randomly selected Interlocking configuration, is simulated with different sample sizes. 
We choose $S=1000$ in the rest of this section, since it leads to a reasonable approximation error for the FC. The FC difference threshold $\delta$ and hold threshold $\Delta$ are instead set to $0.01$ and $5$, respectively, which provide sufficient accuracy to detect when the FC remains constant.
Higher values of $S$ and $\Delta$ and lower values of $\delta$ can improve the prediction accuracy of the minimum unrolling depth at the cost of increased simulation time.

\begin{figure}[t]
    \centering
    \subfigure[]{\includegraphics[width=0.58\columnwidth]{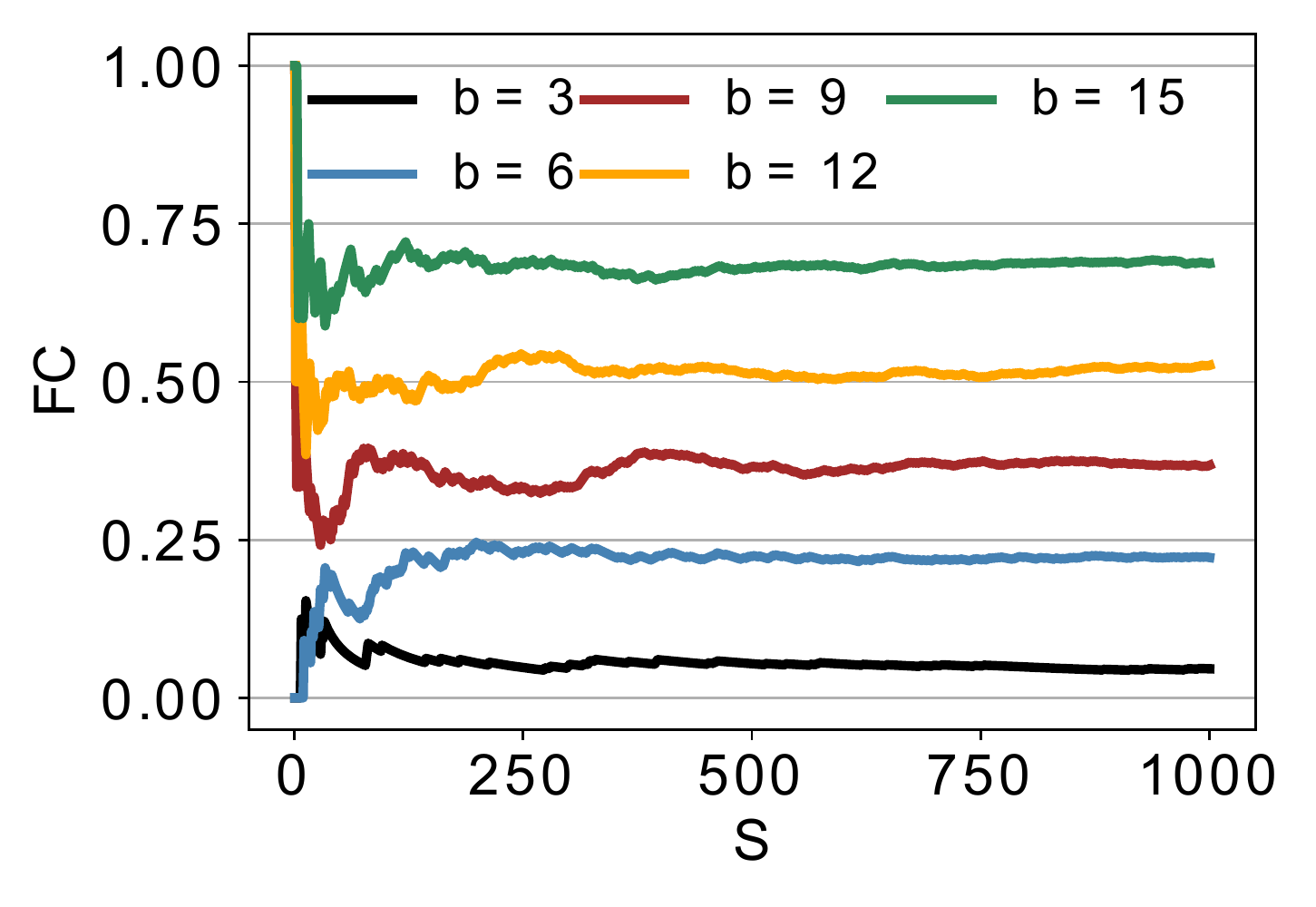}\label{fig:pre_exp}}
    \subfigure[]{\includegraphics[width=0.38\columnwidth]{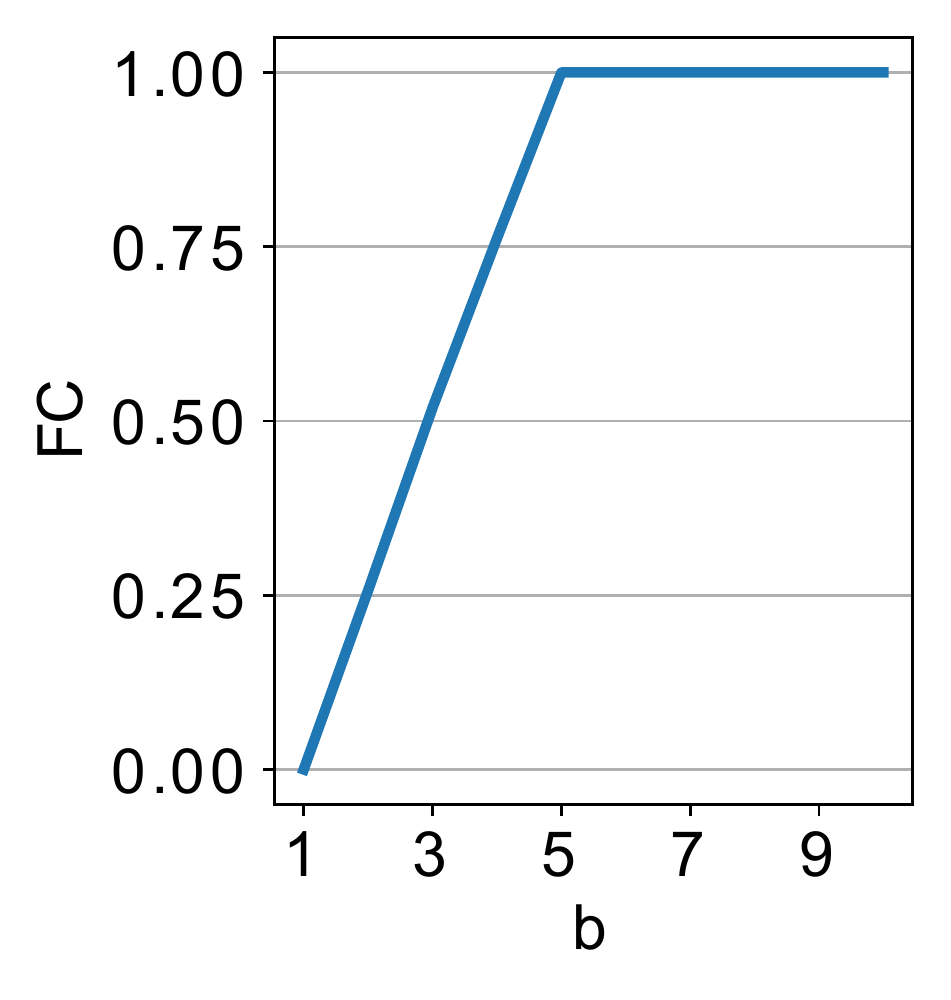}\label{fig:FC_case_study}}
    \caption{(a) FC vs. simulation sample size $S$ and (b) Behavior of FC as a function of the unrolling depth $b$ for an encrypted version of $s38584$.}
    \vspace{-7mm}
\end{figure}

\fakeparagraphnospace{Case Study}
We illustrate \texttt{Fun-SAT} on $s38584$, encrypted using Interlocking with $t_k=4$, corresponding to a key bit-length of 48, and $D_{max}=5$. Fig.~\ref{fig:FC_case_study} shows the behavior of $FC$ for the encrypted circuit as a function of the unrolling depth $b$. 
In the FC analysis phase, the FC is calculated by logic simulation starting from $b=1$. When $b$ is increased to $5$, the FC value starts to stabilize, i.e., the difference between two consecutive FC values is within $\delta$. After a number of additional FC simulations equal to $\Delta$, the first phase terminates by suggesting an unrolling depth of $5$. 
In the SAT-attack phase, \texttt{Fun-SAT} unrolls the encrypted circuit for $5$ cycles and execute the SAT attack on the unrolled circuit. The attack terminates successfully with a single key within $507$~s. 
In contrast, the reference attack in Algorithm~\ref{alg:reference}
starts by unrolling the circuit once ($b=1$) and performs $10$ more checks of the termination conditions (five UK and five MCE checks),  
which takes $2.64$~h and is
$18.7\times$ slower than \texttt{Fun-SAT}. 

\fakeparagraphnospace{Numerical Experiments on the Benchmark Circuits}
We encrypt the first nine selected benchmarks with HARPOON using different configuration parameters. The key size $t_k$ ranges from $1$ to $4$ while the MKF ratio $R_{mkf}$ is assigned four different values, namely, $5\%$, $10\%$, $15\%$, and $20\%$. We then apply both \texttt{Fun-SAT} and the reference attack in Algorithm~\ref{alg:reference}
on all the $144$ encrypted circuits. Fig.~\ref{fig:harpoon_hist} shows the distribution of the attack runtime.

\begin{figure}[t]
    \centering{
    \subfigure[]{
    \includegraphics[width=0.59\columnwidth]{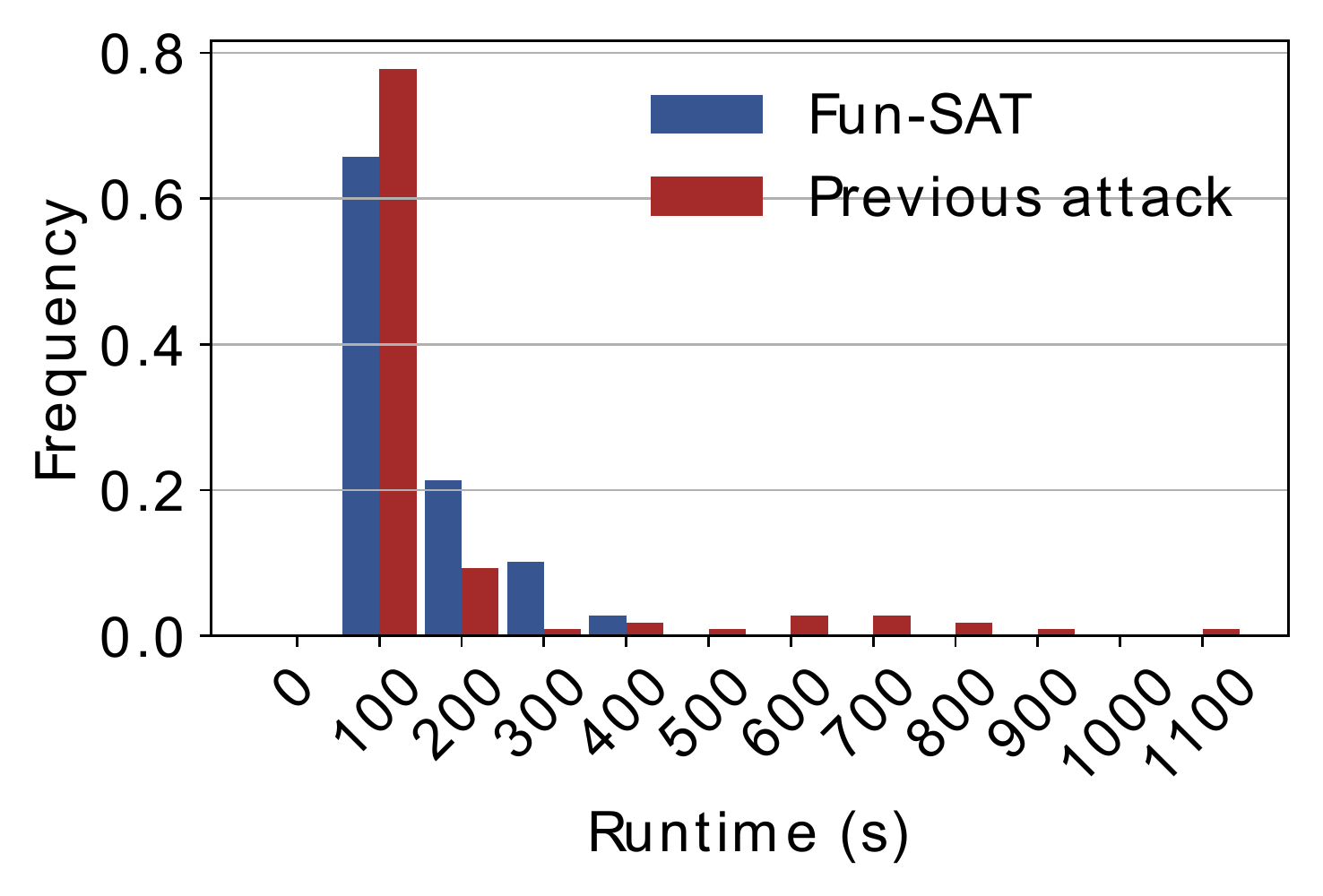}\label{fig:harpoon_hist}}
    \subfigure[]{
    \includegraphics[width=0.36\columnwidth]{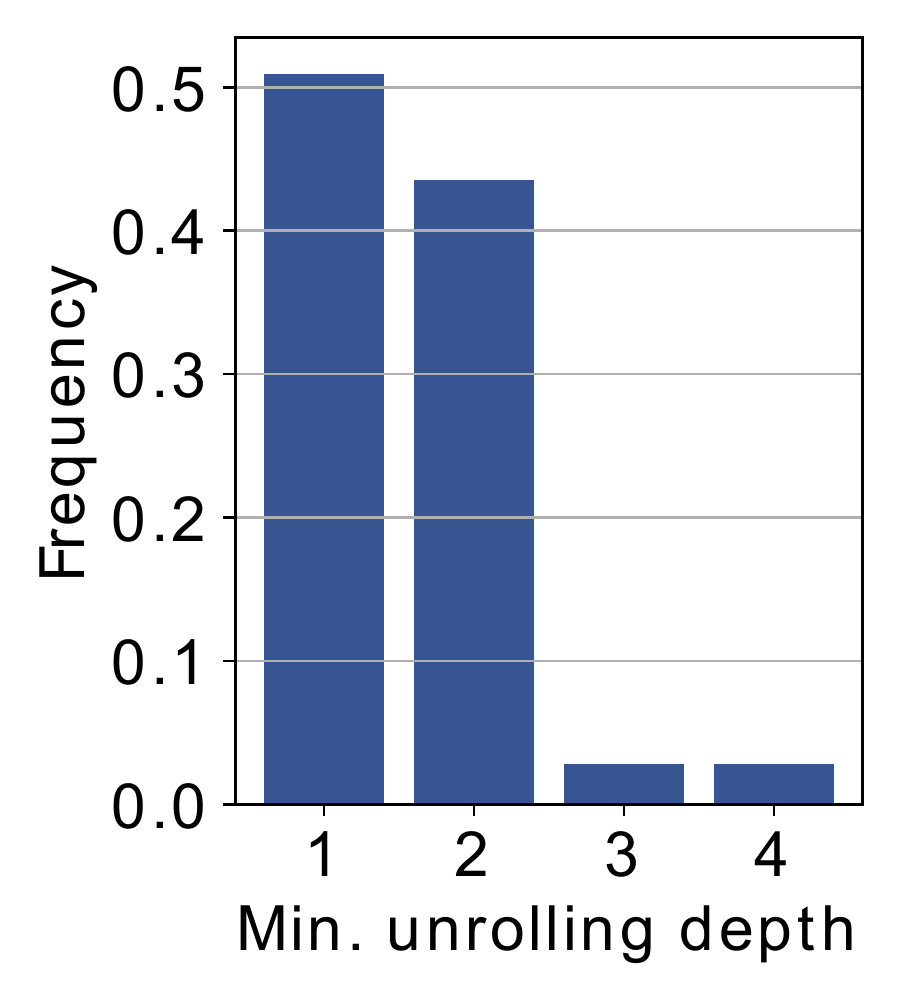}\label{fig:unroll_hist}}
    }
    \vspace{-5mm}
    \caption{Histograms of (a) the attack runtime and (b) the  required unrolling depth for the HARPOON-encrypted circuits.}
    \vspace{-5mm}
\end{figure}

\begin{figure*}[t]
  \centering
  \includegraphics[width=\textwidth]{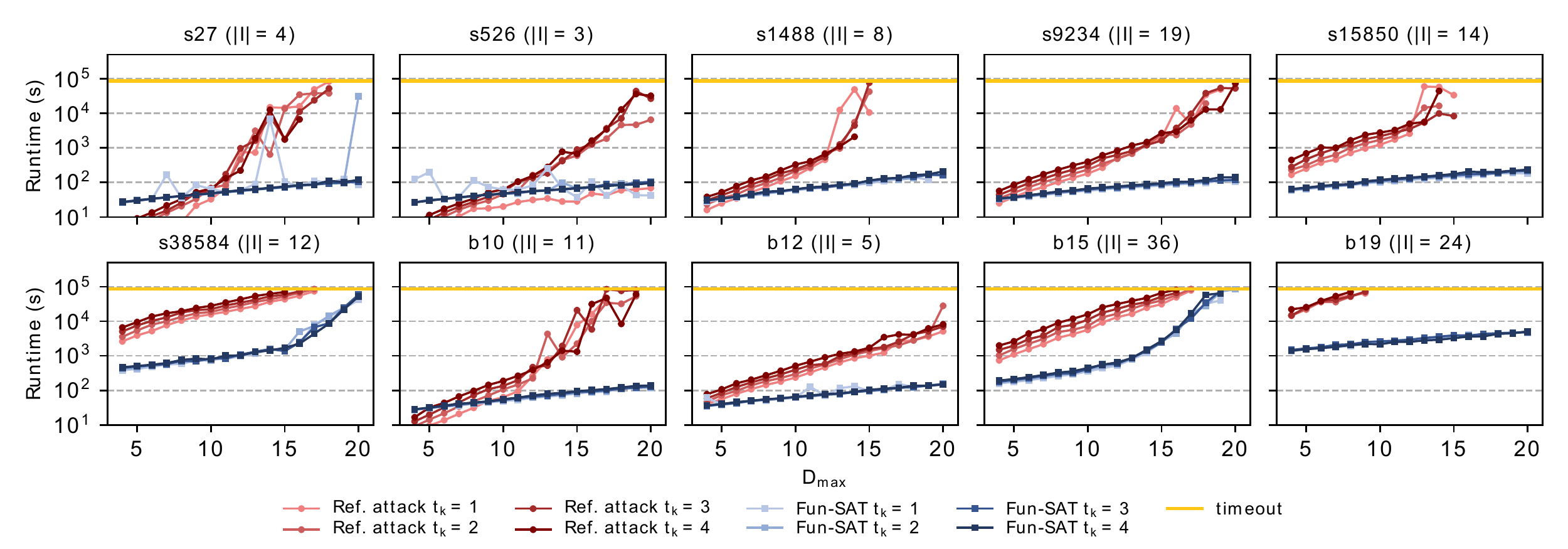}
  \vspace{-8mm}
  \caption{Numerical results for \texttt{Fun-SAT} on Interlocking. 
  (The key bit-length can be computed as $t_k|I|$.)
  }
  \label{fig:interlocking_exp}
  \vspace{-5mm}
\end{figure*}
All the attacks successfully terminate within $20$ minutes, with the average attack time for \texttt{Fun-SAT} and the reference attack being $83$~s and $97$~s, respectively. Such a small average attack runtime for both attacks is mostly due to the small number of unrolling cycles required by HARPOON-encrypted circuits, as shown in Fig.~\ref{fig:unroll_hist}, which are designed to achieve high  functional corruptibility.  
The benchmarks encrypted with Interlocking exhibit, instead, larger unrolling depth  and, therefore, higher resilience to the reference SAT-based attack. We execute the two attacks on circuits encrypted with Interlocking with a time-out threshold of one day. The range of the key size is the same as for the HARPOON configurations while $D_{max}$ ranges from $4$ to $20$. Fig.~\ref{fig:interlocking_exp} shows the attack runtime for all the ten benchmark circuits together with the highlighted time-out threshold. 
The runtime of the reference attack drastically increases with $D_{max}$ for all the benchmarks. Among all the encrypted netlists attacked by the reference attack, $19\%$ reach the time-out after one day and the rest have an average runtime of $7$~hours. In contrast, only $0.7\%$ of the \texttt{Fun-SAT} attempts reach the time-out and $91\%$ terminate successfully within $60$~minutes. On average, \texttt{Fun-SAT}  achieves $90\times$ faster execution than the reference attack whenever both the attacks do not reach time-out. 
Executing \texttt{Fun-SAT} and the reference attack with the $2\cdot b$ update rule on three representative benchmarks, $s1488$, $s15850$, and $s38584$, encrypted with the same Interlocking configuration, still shows   $62\times$, $76\times$, and $10\times$ faster execution than the reference attack, respectively.

\fakeparagraphnospace{Extensions of \texttt{Fun-SAT}} 
We have demonstrated the effectiveness of \texttt{Fun-SAT} in finding the correct key when the key length $t_k$ is known. In the case of an unknown $t_k$, we may still perform the attack by regarding the correct initial state of the circuit, encoded by the connections between the $t_k$-th and the $(t_k+1)$-th circuit replica in Fig.~\ref{fig:after_unroll}, as providing the \emph{effective key} inputs in the unrolled circuit, with no substantial modifications to the circuit representation. This new key effectively configures the correct initial state of the circuit. To access the correct functionality, an attacker can then modify the reset values of all the registers in $C_{e}$ according to the found key.

Recently, sequential encryption techniques~\cite{meade2017revisit,dofe2018novel} have also been proposed that add extra key ports to the encrypted circuit, forcing the users to provide the correct key even if the circuit is in the functional mode. Extending \texttt{Fun-SAT} to address these schemes by accordingly modeling the two types of keys is a possible direction for future research. 
Finally, while the focus of \texttt{Fun-SAT} is on boosting attack efficiency by replacing many, possibly expensive, SAT calls with less expensive logic simulations, the attack may be further improved by using more advanced SAT-attack  tactics, 
such as the incremental SAT-solving in KC2~\cite{shamsi2019kc2}.

\section{Conclusions}\label{sec:conclude}

We presented a functional corruptibility-guided SAT-based attack method that can efficiently estimate the minimum unrolling depth required for a successful SAT attack to prune out of the search space all the wrong keys. \texttt{Fun-SAT} achieves on average two orders of magnitude runtime improvement when compared with a previous reference attack and can effectively be used as a method for evaluating the security of existing sequential encryption schemes. 
Moreover,  \texttt{Fun-SAT} relies on a notion of functional corruptibility for sequential circuits and monotonicity properties that 
are independent of the specific encryption scheme. They can then be applied to accelerate  
other attack variants~\cite{el2017reverse,shamsi2019kc2} or circumvent newly developed encryption schemes~\cite{rezaei2021sequential} that aim to increase the required unrolling depth to achieve resilience against SAT-based attacks. 

\section*{Acknowledgments}
This work was supported in part by the Air Force Research Laboratory (AFRL) and the Defense Advanced Research Projects Agency (DARPA) under agreement number FA8650-18-1-7817.

\bibliography{reference_list} 

\begin{thebibliography}{10}

\bibitem{xiao2015efficient}
K.~Xiao, D.~Forte, and M.~M. Tehranipoor, ``Efficient and secure split
  manufacturing via obfuscated built-in self-authentication,'' in {\em IEEE
  Int. Symp. Hardware Oriented Security and Trust (HOST)}, pp.~14--19, 2015.

\bibitem{yasin2016camoperturb}
M.~Yasin, B.~Mazumdar, O.~Sinanoglu, and J.~Rajendran, ``{CamoPerturb: Secure
  IC camouflaging for minterm protection},'' in {\em 2016 IEEE/ACM Int. Conf.
  Computer-Aided Design (ICCAD)}, pp.~1--8, 2016.

\bibitem{yasin2017evolution}
M.~Yasin and O.~Sinanoglu, ``Evolution of logic locking,'' in {\em IFIP/IEEE
  Int. Conf. Very Large Scale Integration (VLSI-SoC)}, pp.~1--6, IEEE, 2017.

\bibitem{roy2010ending}
J.~A. Roy, F.~Koushanfar, and I.~L. Markov, ``Ending piracy of integrated
  circuits,'' {\em Computer}, vol.~43, no.~10, pp.~30--38, 2010.

\bibitem{rajendran2013fault}
J.~Rajendran, H.~Zhang, C.~Zhang, G.~S. Rose, Y.~Pino, O.~Sinanoglu, and
  R.~Karri, ``Fault analysis-based logic encryption,'' {\em IEEE Trans.
  Computers}, vol.~64, no.~2, pp.~410--424, 2013.

\bibitem{yasin2016improving}
M.~Yasin, J.~J. Rajendran, O.~Sinanoglu, and R.~Karri, ``On improving the
  security of logic locking,'' {\em IEEE Transactions on Computer-Aided Design
  of Integrated Circuits and Systems}, vol.~35, no.~9, pp.~1411--1424, 2016.

\bibitem{chowdhury2021enhancing}
S.~D. Chowdhury, G.~Zhang, Y.~Hu, and P.~Nuzzo, ``Enhancing {SAT}-attack
  resiliency and cost-effectiveness of reconfigurable-logic-based circuit
  obfuscation,'' in {\em Int. Symp. Circuits and Systems (ISCAS)}, pp.~1--5,
  IEEE, 2021.

\bibitem{vivek2019system}
V.~V. Menon, G.~Kolhe, A.~Schmidt, J.~Monson, M.~French, Y.~Hu, P.~A. Beerel,
  and P.~Nuzzo, ``System-level framework for logic obfuscation with quantified
  metrics for evaluation,'' in {\em Secure Development Conf. (SecDev)},
  pp.~89--100, 2019.

\bibitem{hu2019models}
Y.~Hu, V.~V. Menon, A.~Schmidt, J.~Monson, M.~French, and P.~Nuzzo,
  ``Security-driven metrics and models for efficient evaluation of logic
  encryption schemes,'' in {\em ACM-IEEE MEMOCODE}, pp.~1--5, 2019.

\bibitem{patnaik2018best}
S.~Patnaik, M.~Ashraf, O.~Sinanoglu, and J.~Knechtel, ``Best of both worlds:
  Integration of split manufacturing and camouflaging into a security-driven
  {CAD} flow for {3D} ics,'' in {\em Int. Conf. Computer-Aided Design (ICCAD)},
  pp.~1--8, IEEE, 2018.

\bibitem{hu2021risk}
Y.~Hu, K.~Yang, S.~Dutta~Chowdhury, and P.~Nuzzo, ``Risk-aware cost-effective
  design methodology for integrated circuit locking,'' in {\em Design,
  Automation and Test in Europe Conference and Exhibition (DATE)},
  pp.~1182--1185, IEEE, 2021.

\bibitem{mohan2021hardware}
P.~Mohan, O.~Atli, J.~Sweeney, O.~Kibar, L.~Pileggi, and K.~Mai, ``Hardware
  redaction via designer-directed fine-grained {eFPGA} insertion,'' in {\em
  Design, Automation \& Test in Europe Conference \& Exhibition (DATE)},
  pp.~1186--1191, IEEE, 2021.

\bibitem{chakraborty2009harpoon}
R.~S. Chakraborty and S.~Bhunia, ``{HARPOON}: An obfuscation-based {SoC} design
  methodology for hardware protection,'' {\em IEEE Trans. Computer-Aided Design
  of Integrated Circuits and Systems}, vol.~28, no.~10, pp.~1493--1502, 2009.

\bibitem{desai2013interlocking}
A.~R. Desai, M.~S. Hsiao, C.~Wang, L.~Nazhandali, and S.~Hall, ``Interlocking
  obfuscation for anti-tamper hardware,'' in {\em Proc. Cyber Security and
  Information Intelligence Research Workshop}, pp.~1--4, 2013.

\bibitem{subramanyan2015evaluating}
P.~Subramanyan, S.~Ray, and S.~Malik, ``Evaluating the security of logic
  encryption algorithms,'' in {\em IEEE Int. Symp. Hardware Oriented Security
  and Trust (HOST)}, pp.~137--143, 2015.

\bibitem{wang2017secure}
X.~Wang, D.~Zhang, M.~He, D.~Su, and M.~Tehranipoor, ``Secure scan and test
  using obfuscation throughout supply chain,'' {\em IEEE Trans. Computer-Aided
  Design of Integrated Circuits and Systems}, vol.~37, no.~9, pp.~1867--1880,
  2017.

\bibitem{shamsi2019kc2}
K.~Shamsi, M.~Li, D.~Z. Pan, and Y.~Jin, ``{KC2}: Key-condition crunching for
  fast sequential circuit deobfuscation,'' in {\em Design, Automation and Test
  in Europe Conference and Exhibition (DATE)}, pp.~534--539, 2019.

\bibitem{el2017reverse}
M.~El~Massad, S.~Garg, and M.~Tripunitara, ``Reverse engineering camouflaged
  sequential circuits without scan access,'' in {\em 2017 IEEE/ACM
  International Conference on Computer-Aided Design (ICCAD)}, pp.~33--40, 2017.

\bibitem{meade2017revisit}
T.~Meade, Z.~Zhao, S.~Zhang, D.~Pan, and Y.~Jin, ``Revisit sequential logic
  obfuscation: Attacks and defenses,'' in {\em IEEE Int. Symp. Circuits and
  Systems (ISCAS)}, pp.~1--4, 2017.

\bibitem{hu2020sanscrypt}
Y.~Hu, K.~Yang, S.~Nazarian, and P.~Nuzzo, ``{SANSCrypt}: A
  sporadic-authentication-based sequential logic encryption scheme,'' in {\em
  IFIP/IEEE Int. Conf. Very Large Scale Integration (VLSI-SoC)}, pp.~129--134,
  2020.

\bibitem{hu2020sanscrypt_extended}
Y.~Hu, K.~Yang, S.~Nazarian, and P.~Nuzzo, ``{SANSCrypt}:
  Sporadic-authentication-based sequential logic encryption,'' in {\em
  VLSI-SoC: Design Trends} (A.~Calimera, P.-E. Gaillardon, K.~Korgaonkar,
  S.~Kvatinsky, and R.~Reis, eds.), (Cham), pp.~255--278, Springer
  International Publishing, 2021.

\bibitem{dofe2018novel}
J.~Dofe and Q.~Yu, ``Novel dynamic state-deflection method for gate-level
  design obfuscation,'' {\em IEEE Trans. Computer-Aided Design of Integrated
  Circuits and Systems}, vol.~37, no.~2, pp.~273--285, 2018.

\bibitem{sorensson2005minisat}
N.~S{\"o}rensson and N.~E{\'e}n, ``Minisat—a {SAT} solver with
  conflict-clause minimization,'' in {\em Eighth International Conference on
  Theory and Applications of Satisfiability Testing (SAT)}, vol.~3569, 2005.

\bibitem{mcmillan2003interpolation}
K.~L. McMillan, ``Interpolation and {SAT}-based model checking,'' in {\em Int.
  Conf. Computer Aided Verification}, pp.~1--13, Springer, 2003.

\bibitem{cavada2014nuxmv}
R.~Cavada, A.~Cimatti, M.~Dorigatti, A.~Griggio, A.~Mariotti, A.~Micheli,
  S.~Mover, M.~Roveri, and S.~Tonetta, ``The {nuXmv} symbolic model checker,''
  in {\em Int. Conf. Computer Aided Verification}, pp.~334--342, Springer,
  2014.

\bibitem{brglez1989combinational}
F.~Brglez, D.~Bryan, and K.~Kozminski, ``Combinational profiles of sequential
  benchmark circuits,'' in {\em IEEE Int. Symp. Circuits and Systems (ISCAS)},
  pp.~1929--1934, 1989.

\bibitem{davidson1999itc}
S.~Davidson, ``{ITC'99} benchmark circuits-preliminary results,'' in {\em
  International Test Conference 1999. Proceedings (IEEE Cat. No. 99CH37034)},
  pp.~1125--1125, IEEE, 1999.

\bibitem{rezaei2021sequential}
A.~Rezaei and H.~Zhou, ``Sequential logic encryption against model checking
  attack,'' in {\em Design, Automation and Test in Europe Conference and
  Exhibition (DATE)}, pp.~1178--1181, IEEE, 2021.

\end{thebibliography}
\bibliographystyle{ieeetr}
\end{document}